\documentclass[11pt,a4paper]{article}
\usepackage[margin=1in]{geometry}
\usepackage{amsmath}
\usepackage{amssymb}
\usepackage{amsthm}
\usepackage{mathrsfs}
\usepackage{amsbsy}
\usepackage{upgreek}
\usepackage{enumerate}
\usepackage{todonotes}
\usepackage[normalem]{ulem}%

\newtheorem{definition}{Definition}
\newtheorem{theorem}{Theorem}
\newtheorem{lemma}{Lemma}
\newtheorem{problem}{Problem}

\DeclareMathOperator{\opt}{OPT}
\DeclareMathOperator{\alg}{ALG}

\newcommand{\Sep}{\ensuremath\mathcal{X}}
\newcommand{\IR}{\mathbb{R}}
\newcommand{\mb}[1]{\mathbf{#1}}

\usepackage{authblk}
\usepackage[unicode=true]{hyperref}
\newcommand\rurl[1]{\href{http://#1}{\nolinkurl{#1}}}

\title{Approximation Schemes for Geometric Coverage Problems}
\author[1]{Steven Chaplick}
\author[2]{Minati De}
\author[3]{Alexander Ravsky}
\author[1]{Joachim Spoerhase}

\affil[1]{Lehrstuhl f\"ur Informatik~I, Universit\"at W\"urzburg, Germany,
  \rurl{www1.informatik.uni-wuerzburg.de/en/staff},
  \nolinkurl{firstname.lastname@uni-wuerzburg.de}.}
\affil[2]{Department of Computer Science and Automation,
Indian Institute of Science, Bangalore, India
\nolinkurl{minati@csa.iisc.ernet.in}.}
\affil[3]{Pidstryhach Institute for Applied Problems of Mechanics and Mathematics, National Academy of Science of Ukraine, Lviv, Ukraine,
  \nolinkurl{oravsky@mail.ru}.}

\date{\today}

\begin{document}
\maketitle

\begin{abstract}
  In their seminal work,
  Mustafa and Ray~\cite{MustafaR09} showed that a wide class of
  geometric \emph{set cover (SC)} problems admit a PTAS via local
  search -- this is one of the most general approaches known for
  such problems. Their result applies if a naturally defined
  ``exchange graph'' for two feasible solutions is planar and is based
  on subdividing this graph via a planar separator theorem due to
  Frederickson~\cite{Frederickson87}. Obtaining similar results for
  the related \emph{maximum $k$-coverage problem (MC)} seems
  non-trivial due to the hard cardinality constraint. In fact, while
  Badanidiyuru, Kleinberg, and Lee~\cite{badanidiyuru12-geometric-mc}
  have shown(via a different analysis) that local search yields a PTAS for
  two-dimensional real halfspaces, they only conjectured that the same
  holds true for dimension three. Interestingly, at this point
  it was already known that local search provides a PTAS for the
  corresponding set cover case and this followed directly from
  the approach of Mustafa and Ray.

  In this work we provide a way to address the above-mentioned
  issue. First, we propose a \emph{color-balanced} version of the
  planar separator theorem. The resulting subdivision approximates
  locally in each part the global distribution of the colors. Second,
  we show how this roughly balanced subdivision can be employed in a
  more careful analysis to strictly obey the hard cardinality
  constraint. More specifically, we obtain a PTAS for any
  ``planarizable'' instance of MC and thus essentially for all cases
  where the corresponding SC instance can be tackled via the approach
  of Mustafa and Ray. As a corollary, we confirm the conjecture of
  Badanidiyuru, Kleinberg, and Lee~\cite{badanidiyuru12-geometric-mc}
  regarding real half spaces in dimension three. We feel that our
  ideas could also be helpful in other geometric settings involving a
  cardinality constraint.
\end{abstract}

\section{Introduction}

The Maximum Coverage (MC) problem is one of the classic combinatorial optimization problems which is well studied  due to its wealth of applications.
 Let $U$ be a set of ground elements, $\mathcal{F}\subseteq 2^{U}$ be a family of subsets of $U$ and $k$ be a positive integer.  The \textsc{Maximum Coverage} (MC) problem asks for a $k$-subset $\mathcal{F}'$ of $\mathcal{F}$ such that the number $|\bigcup\mathcal{F}'|$ of ground elements covered by $\mathcal{F}'$ is maximized.

Many real life problems arising from banking~\cite{CornuejolsNW80}, social networks,
transportation network~\cite{Mecke}, databases~\cite{HarinarayanRU96}, information
retrieval, sensor placement, security (and others) can be framed as an instance of
MC problem. For example, the following are easily seen as MC problems: placing $k$
sensors to maximize the number of covered customers, finding a set of $k$ documents
satisfying the information needs of as many users as possible~\cite{badanidiyuru12-geometric-mc},
and placing $k$ security personnel in a terrain to maximize the number of secured regions is secured.

  From the result of Cornu{\'{e}}jols~\cite{CornuejolsNW80}, it is
  well known that greedy algorithm is a $1-1/e$
  approximation algorithm for the MC problem.  Due to wide applicability of the
  problem, whether one can achieve an approximation factor better than
  $(1-\frac{1}{e})$ was subject of research for a long period of time.
  From the result of Feige~\cite{Feige98}, it is known that if there
  exists a polynomial-time algorithm that approximates maximum
  coverage within a ratio of $(1 - \frac{1}{e} + \epsilon)$ for some
  $\epsilon > 0$ then P = NP.  Better results can however be
  obtained for special cases of MC.  For example, Ageev and Sviridenko
  \cite{ageevS04-pipage-rounding} show in their seminal work that
  their pipage rounding approach gives a factor $1-(1-1/r)^r$ for
  instances of MC where every element occurs in at most $r$ sets.  For
  constant $r$ this is a strict improvement on $1-1/e$ but this bound
  is approached if $r$ is unbounded.  For example, pipage rounding
  gives a $3/4$-approximation algorithm for \textsc{Maximum Vertex
    Cover} (MVC), which asks for a $k$-subset of nodes of a given
  graph that maximizes the number of edges incident on at least one of
  the selected nodes.  Petrank~\cite{Patrank94} showed that this
  special case of MC is APX-hard.

In this paper, we study the approximability of MC in \emph{geometric} settings where elements and sets are represented by geometric objects.  Such problems have been considered before and have applications, for example, in information retrieval~\cite{badanidiyuru12-geometric-mc} and in wireless networks~\cite{Erlebach08-geometric-coverage}.

MC is related to the \textsc{Set Cover} problem (SC). For a given set  $U$  of ground elements and  a family   $\mathcal{F}\subseteq 2^{U}$ of subsets of $U$, this problem asks for a minimum cardinality subset of $\mathcal{F}$ which covers all the ground elements of $U$. This problem plays a central role in combinatorial optimization and in particular in the study of approximation algorithms. The best known approximation algorithm has a ratio of $\ln n$, which is essentially the best possible \cite{Feige98} under a plausibly complexity-theoretic assumption. A lot of work has been devoted to beat the logarithmic barrier in the context of geometric set cover problems\cite{BronnimannG95-geometric-setcover,Varada,Chan,Ray}.  Mustafa and Ray~\cite{MustafaR09} introduced a powerful tool which can be used to show that a \emph{local search} approach provides a PTAS for various geometric SC problems.  Their result applies if a naturally defined ``exchange graph'' (whose nodes are the sets in two feasible solutions) is planar and is based on subdividing this graph via a planar separator theorem due to Frederickson~\cite{Frederickson87}.
In the same paper~\cite{MustafaR09}, they applied this approach to provide a PTAS for the SC problem when the family $\mathcal{F}$ consists of either a set of half spaces in $\IR^3$, or a set of disks in $\IR^2$.
Many results have been obtained using this technique for different problems in geometric settings~\cite{ChanH12,DeL16,GibsonP10,KrohnGKV14}. Some of these works extend to cases where the underlying exchange graph is not planar but admits a small-size separator~\cite{AschnerKMY13,GovindarajanRRR16,Har-PeledQ15}.

Beyond the context of SC, local search has also turned out to be a very
powerful tool  for other geometric problems but the analysis of
such algorithms is usually non-trivial and highly tailored to the
specific setting. Examples are the Euclidean TSP, Euclidean Steiner
tree, facility location, $k$-median~\cite{Cohen-AddadM15}. In some very
recent breakthroughs,  PTASs for $k$-means problem in finite
Euclidean dimension (and more general cases) via local search have
been announced
\cite{cohen-addad-local-search-k-means,friggstad-etal-local-search-k-means}.

In this paper, we study the effectiveness of local search for geometric MC problems.  In the general case, $b$-swap local search is known to yield a tight approximation ratio of $1/2$ \cite{KerkkampA16-localsearch-maxcoverage}.  However, for special cases such as geometric MC problems local search is a promising candidate for beating the barrier $1-1/e$. It seems, however, non-trivial to obtain such results using the technique of Mustafa and Ray~\cite{MustafaR09}. In their analysis, each part of the subdivided planar exchange graph (see above) corresponds to a feasible candidate swap that replaces some sets of the local optimum with some sets of the global optimum and it is ensured that every element stays covered due to the construction of the exchange graph.  It is moreover argued that if the global optimum is sufficiently smaller than the local optimum then one of the considered candidate swaps would actually reduce the size of the solution.

It is possible to construct the same exchange graphs also for the case
of MC. However, the hard cardinality constraint given by input
parameter $k$ poses an obstacle. In particular, when considering a
swap corresponding to a part of the subdivision, this swap might
be infeasible as it may contain (substantially) more sets from the
global optimum than from the local optimum. Another issue is that MC
has a different objective function than SC. Namely, the goal is to maximize the
number of covered elements rather than minimizing the number of used
sets.  Finally, while for SC all elements are covered by both
solutions, in MC we additionally have elements that are covered by
none or only one of the two solutions requiring a more detailed distinction of several types of elements.

In fact, subsequent to the work of Mustafa and Ray on
SC~\cite{MustafaR09}, Badanidiyuru, Kleinberg, and
Lee~\cite{badanidiyuru12-geometric-mc} studied geometric MC.
They obtained fixed-parameter approximation schemes for MC instances for
the very general case where the family $\mathcal{F}$ consists of
objects with bounded VC dimension, but the running times are
exponential in the cardinality bound $k$. They further provided APX-hardness for each of the
following cases: set systems of VC-dimension 2, halfspaces in $\IR^4$,
and axis-parallel rectangles in $\IR^2$.
Interestingly, while they have shown that for MC instances where
$\mathcal{F}$ consists of halfspaces in $\IR^2$ local search can be
used to provide a PTAS, they only conjecture that local search will provide
a PTAS for when $\mathcal{F}$ consists of half spaces in $\IR^3$.
This underlines the observation that it seems non-trivial to apply the approach of Mustafa and Ray to geometric MC
problems as at that point a PTAS for halfspaces in $\IR^3$ for SC was already
known via the approach of Mustafa and Ray.

The difficulty of analyzing local search under the presence of a cardinality constraint is also known in other settings. For example, one of the main technical contributions of the recent breakthrough for the Euclidean $k$-means problem \cite{cohen-addad-local-search-k-means,friggstad-etal-local-search-k-means} is that the authors are able to handle the hard cardinality constraint by the concept of so-called isolated pairs \cite{cohen-addad-local-search-k-means}.  Prior to these works approximation schemes have only been known for bicriteria variants where the cardinality constraint may be violated or where there is no constraint but---analogously to SC---the cardinality contributes to the objective function~\cite{BandyapadhyayV16-kmeans-clustering}.

\subsection{Our Contribution}

In this paper, we show a way how to cope with the above-mentioned
issue with a cardinality constraint.  We are able to achieve a PTAS
for many geometric MC problems. At a high level we follow the
framework of Mustafa and Ray defining a planar (or more generally
$f$-separable) exchange graph and subdividing it into a number of
small parts each of them corresponding to a candidate swap. As each
part may be (substantially) imbalanced in terms of the number of sets
of the global optimum and local optimum, respectively, a natural idea
seems to swap in only a sufficiently small subset of the globally
optimal sets.  This idea alone is, however, not sufficient.  Consider,
for example, the case where each part contains either only sets from
the local or only sets from the global optimum making it impossible to
retrieve any feasible swap from the considering the single parts. To
overcome this difficulty, we prove in a first step a
\emph{color-balanced} version of the planar separator theorem
(Theorem~\ref{thm:2color-uniform}). In this theorem, the input is a
planar (or more generally \emph{$f$-separable}) graph whose nodes are
two-colored arbitrarily. The distinctions of our separator theorem
from the prior work, are that our separator theorem guarantees that
all parts have roughly the same size (rather than simply an upper
limit on their size) and that the two colors are represented in each
part in roughly the same ratio as in the whole graph.  This balancing
property allows us to address the issue of the above-mentioned
infeasible swaps. In a second step, we are able to employ the only
roughly color-balanced subdivision to establish a set of perfectly
balanced candidate swaps. We prove by a careful analysis (which turns
out more intricate than for the SC case) that local search also yields
a PTAS for the wide class \emph{$f$-separable} MC problems (see
Theorem~\ref{thm:main-thm}). As an immediate consequence, we obtain
PTASs for essentially all cases of geometric MC problems where the
corresponding SC problem can be tackled via the approach of Mustafa
and Ray (Theorem~\ref{thm:Application}). In particular, this confirms
the conjecture of Badanidiyuru, Kleinberg, and
Lee~\cite{badanidiyuru12-geometric-mc} regarding halfspaces in
$\IR^3$.  We also immediately obtain PTASs for \textsc{Maximum
  Dominating Set} and \textsc{Maximum Vertex Cover} on $f$-separable
and minor-closed graph classes (see section~\ref{sec:Application}),
which, to the best of our knowledge, were not known before.  We feel
that our approach has the potential to find further applications in
similar cardinality constrained settings.

\section{Color Balanced Divisions}
\label{sec:tools}

In this section we provide the main tool used to prove our main result (i.e.,
Theorem~\ref{thm:main-thm}).
We first describe a new subtle specialization (see
Lemma~\ref{lem:uniform-(r,f)-division}) of the standard division theorem
on $f$-separable graph classes (see Theorem~\ref{thm:(r,f)-division}).
This builds on the concept of $(r,f(r))$-divisions (in the sense of
Henzinger et al.~\cite{HenzingerKRS97}) of graphs in an $f$-separable graph
class. We then extend this specialized division lemma by suitably aggregating
the pieces of the partition to obtain a \emph{two-color balanced} version (see
Theorem~\ref{thm:2color-uniform}).
This result generalizes to more than two colors.
However, as our applications stem from the two-colored version, we defer the
generalization to the appendix (see Appendix~\ref{app:separators}).
For a number $n$, we use $[n]$ to denote the set $\{1,\dots,n\}$.

  For a graph $G$, a subset $S$ of $V(G)$ is an \emph{$\alpha$-balanced separator}
when its removal breaks $G$ into two collections of connected components
such that each collection contains at most an $\alpha$ fraction of
$V(G)$ where $\alpha \in [\frac{1}{2},1)$ and $\alpha$ is a constant.
The size of a separator $S$ is simply the number of vertices it contains.
For a non-decreasing sublinear function $f$, a class of graphs that is closed under taking subgraphs is said to be \emph{$f$-separable}
 if there is an $\alpha \in [\frac{1}{2},1)$ such that for any $n>2$, an $n$-vertex graph in the class has a $\alpha$-balanced separator whose size is at most $f(n)$.
Note that, by the Lipton-Tarjan separator theorem~\cite{lipton1979separator},
planar graphs are a subclass of the $\sqrt{n}$-separable graphs. More generally,
Alon, Seymour, and Thomas~\cite{alonST1990minorclosed} have shown that every
graph class characterized by a finite set of forbidden minors is also a subclass
of the $(c\cdot\sqrt{n})$-separable graphs (here, the constant $c$ depends on the size
of the largest forbidden minor). In particular, from the graph minors
theorem~\cite{robertson2004graph}, every non-trivial minor closed graph class is
a subclass of the $(c\cdot\sqrt{n})$-separable graphs (for some constant $c$).
Note that when we discuss $f$-separable graph classes we assume the function $f$ has the
form $f(x) = x^{1-\delta}$ for some $\delta >0$, i.e., it is both non-decreasing and
strongly sublinear.

Frederickson~\cite{Frederickson87} introduced the notion of an
\emph{$r$-division} of an $n$-vertex graph $G$, namely, a cover
of $V(G)$ by $\Theta(\frac{n}{r})$ sets each of size $O(r)$ where each set
has $O(\sqrt{r})$ \emph{boundary} vertices, i.e., $O(\sqrt{r})$ vertices in
common with the other sets. Frederickson
showed that, for any $r$, every planar graph $G$ has an $r$-division and
that one can be computed in $O(n \log n)$ time. This result follows from
a recursive application of the Lipton-Tarjan planar separator
theorem~\cite{lipton1979separator}. This notion was further generalized by
Henzinger et al.~\cite{HenzingerKRS97} to $(r,f(r))$-divisions\footnote{They
use a more general notion of $(r,s)$-division but we need the restricted
version as described here.} where
$f$ is a function in $o(r)$ and each set has at most $f(r)$ vertices in
common with the other sets. They noted that Frederickson's proof can
easily be adapted to obtain an $(r,c \cdot f(r))$-division of any graph $G$
from a subgraph closed $f$-separable graph class -- as formalized in
Theorem~\ref{thm:(r,f)-division}).
Note that we use an equivalent but slightly different notation than
Frederickson and Henzinger et al. in that we consider the ``boundary'' vertices
as a single separate set apart from the non-boundary vertices in each ``region'',
i.e., our \emph{divisions} are actually partitions of the vertex set.
This allows us to carefully describe the number of vertices inside each ``region''.

\begin{theorem}[\cite{Frederickson87,HenzingerKRS97}]
\label{thm:(r,f)-division}
For any subgraph closed $f$-separable
class of graphs $\mathcal{G}$, there are constants $c_1,c_2$ such that every
graph $G$ in the class has an \emph{$(r, c_1 \cdot f(r))$-division} for any $r$.
Namely, for any $r\geq 1$, there is an integer $t \in \Theta(\frac{n}{r})$
such that $V$ can be partitioned into $t+1$ sets $\Sep, V_1, \ldots, V_t$
where the following properties hold.
\begin{enumerate}[(i)]
  \item $N(V_i) \cap V_j= \emptyset$ for each $i\neq j$,
  \item $|V_i \cup N(V_i)| \leq r$ for each $i$,
  \item $|N(V_i) \cap \Sep| \leq c_1 \cdot f(r)$ for each $i$
  (thus, $|\Sep| \leq \sum_{i=1}^t |\Sep \cap N(V_i)| \leq
         c_2 \cdot \frac{f(r) \cdot n}{r}$).
  \end{enumerate}
Moreover, such a partition can be found in $O(g(n) \log n)$ time where
$g(n)$ is the time required to find an $f$-separation in $\mathcal{G}$.
\end{theorem}

We specialize the notion of $(r,f(r))$-divisions first to \emph{uniform}
$(r,f(r))$-divisions, and then generalize to \emph{two-color uniform
$(r,f(r))$-divisions} of a two-colored graph (note: the coloring need not be
proper in the usual sense).
A \emph{uniform $(r,f(r))$-division} is an $(r,f(r)$-division where
the $\Theta(\frac{n}{r})$ sets have a \emph{uniform} (i.e., $\Theta(r)$)
amount of \textbf{internal} vertices.
A \emph{two-color uniform $(r,f(r))$-division} of a two-colored graph is a
uniform $(r,f(r))$-division where each set additionally has the ``same''
proportion of each color class (this is formalized in
Theorem~\ref{thm:2color-uniform}).

It is important to note that while this uniformity condition (i.e., that each 
region is not \emph{too small}) has not been needed in the 
past\footnote{E.g., to analyse local search for SC 
problems~\cite{MustafaR09}, or for fast algorithms to find shortest 
paths~\cite{HenzingerKRS97}.}, it is essential for our analysis of local 
search as applied to MC problems in the next section.
Moreover, to the best of our knowledge, neither Frederickson's construction 
nor more modern constructions (e.g.~\cite{KleinMS13}) of an $r$-division 
explicitly guarantee that the resulting $r$-division is uniform.
To be specific, Frederickson's approach consists of two steps. The first step 
recursively applies the separator theorem until each region together with its 
boundary is ``small enough''. In the second step, each
region where the boundary is ``too large'' is further divided. This is 
accomplished applying the separator theorem to a weighted version of each 
such region where the boundary vertices are uniformly weighted and the non-
boundary vertices are zero-weighted.
Clearly, even a single application of this latter step may result in regions 
with $o(r)$ interior vertices.
Modern approaches (e.g.~\cite{KleinMS13}) similarly involve applying the separator theorem to weighted regions where boundary vertices are uniformly weighted and interior vertices are zero-weighted, i.e., regions which are \emph{too small} are not explicitly avoided.

The remainder of this section is outlined as follows.
We will first show for every $f$-separable graph class $\mathcal{G}$ there is 
a constant $c$ such that every graph in $\mathcal{G}$
has a uniform $(r, c \cdot f(r))$-division (see Lemma~\ref{lem:uniform-(r,f)-division}).
We then use this result to show that for every $f$-separable graph class
$\mathcal{G}$ there is a constant $c'$ such that every two-colored graph in
$\mathcal{G}$ has a two-color uniform $(rq, c' \cdot q \cdot
f(r))$-division for any $q$ -- see Theorem~\ref{thm:2color-uniform}.
Our proofs are constructive and lead to efficient algorithms which produce
such divisions when there is a corresponding efficient algorithm to compute
an $f$-separation.

To prove the first result, we start from a given $(r,f(r))$-division and
``group'' the sets carefully so that we obtain the desired uniformity. For
the two-colored version, we start from a uniform $(r,f(r))$-division and again regroup the sets via a reformulation of the problem as a partitioning
problem on two-dimensional vectors. Namely, we leverage Lemma~\ref{lem:partitioning_easy} to perform the regrouping.

\begin{lemma}\label{lem:uniform-(r,f)-division}
Let $\mathcal{G}$ be a $f$-separable graph class and $G=(V,E)$ be a sufficiently large
$n$-vertex graph in $\mathcal{G}$. There are constants $r_0, x_0$ (depending only on $f$) such that for any $r \in [r_0, \frac{n}{x_0}]$ there is an integer
$t \in \Theta(\frac{n}{r})$
such that $V$ can be partitioned into $t+1$ sets
$\Sep, V_1, \ldots, V_t$ where $c_1,c_2$ are constants independent
of $r$ and the following properties are satisfied.
\begin{enumerate}[(i)]
  \item $N(V_i) \cap V_j= \emptyset$ for each $i\neq j$,
  \item $|V_i| \in [\frac{r}{2},2r]$ for each $i$,
  \item $|N(V_i) \cap \Sep| \leq c_1 \cdot f(r)$ for each $i$
  (thus, $|\Sep| \leq \sum_{i=1}^t|\Sep \cap N(V_i)| \leq \frac{c_2 \cdot f(r) \cdot n}{r}$).
\end{enumerate}
Moreover, such a partition can be found in $O(h(n) + n)$ time where $h(n)$ is
the amount of time required to produce an $(r,f(r))$-division of $G$.
\end{lemma}
\begin{proof}
We start from an $(\lfloor \frac{r}{8}\rfloor, c_1 \cdot f(\lfloor \frac{r}{8}\rfloor ))$-division
$\mathcal{U} = (\Sep, U_1, \dots, U_\ell)$ as given by
Theorem~\ref{thm:(r,f)-division}
where $\ell = c_\ell \cdot \frac{8\cdot n}{r}$. We then partition $[\ell]$
into $t$ sets $I_1, \dots, I_t$ such that $(\Sep, V_1, \dots, V_t)$ is a
uniform $(r, c \cdot f(r))$-division $\Sep, V_1, \dots, V_t$ where
$V_i = \bigcup_{j \in I_i} U_j$. In order to describe the partitioning, we
first observe some useful properties of $U_1, \dots, U_\ell$ where,
without loss of generality, $|U_1| \geq \dots \geq |U_\ell|$.
Let $n^* = \sum_{j=1}^{\ell} |U_j|$, and set $t = \lceil\frac{n^*}{r}\rceil$. Note that:
\begin{equation}
n^* = \sum_{j=1}^{\ell} |U_j| = n - |\Sep| \geq
n \cdot \left(1 - \frac{c_2 \cdot f(\lfloor\frac{r}{8}\rfloor)}{\lfloor\frac{r}{8}\rfloor}\right).
\end{equation}
From our choice of $t$, the average size of the sets $V_i$ is $\frac{n^*}{t} \in (\frac{r}{1+\frac{r}{n^*}}, r]$.

Pick $r_0$ such that it is divisible by 8 and $c^* = 1 - c_2 \cdot f(\frac{r_0}{8})\cdot (\frac{r_0}{8})^{-1}>0$ and assume $r\geq r_0$ in what follows.
Then $n^* \geq c^* \cdot n$, i.e., $c^* \leq \frac{n^*}{n}$. Now pick $x_0 = 
\frac{3}{c^*}$. Thus, we have $r \leq \frac{n}{x_0} \leq \frac{n^*}{3}$. In 
particular, the average size of our sets $|V_i|$ is in $[\frac{3r}{4},r]$.

Notice that $\frac{\ell}{t} \leq c_\ell \cdot \frac{8\cdot n}{r} \cdot (\frac{n^*}{r})^{-1}
\leq \frac{8 c_\ell}{c^*}$.
We build the sets $I_i$ such that $|I_i| \leq 40 \cdot \frac{c_\ell}{c^*}$.
This provides $|N(V_i) \cap \Sep| \leq 40\cdot \frac{c_\ell}{c^*} \cdot
c_1 f(\lfloor \frac{r}{8} \rfloor) \in O(f(r))$.

We build the sets $I_i$ in two steps.
In the first step we greedily fill the sets $I_i$ according to the largest
unassigned set $U_j$ (formalized as follows).
For each $j^*$ from $1$ to $\ell$, we consider an index $i^* \in [t]$ where
$|I_{i^*}| < 32 \cdot \frac{c_\ell}{c^*}$ and $|V_{i^*}|$ is minimized.
If $|V_{i^*}| \leq \frac{n^*}{t}$, then we place $j^*$ into $I_{i^*}$, that is,
we replace $V_{i^*}$ with $V_{i^*} \cup U_{j^*}$. Otherwise (there is no such index $i^*$),
we proceed to step two (below). Before discussing step two, we first
consider the state of the sets $V_i$ at the moment when this greedy
placement finishes. To this end, let $j^*$ be the index of the first
(i.e., the largest) $U_j$ which has not been placed.

\smallskip

\noindent\textbf{Claim 1:} \textit{If $|V_i| \leq \frac{n^*}{t}$ for every $i$,
then all each set $U_j$ has been merged into some $V_i$ and the $V_i$'s satisfy the conditions of the lemma.} \\
First, suppose there is an unallocated set $U_j$.
Since $|V_i| \leq \frac{n^*}{t}$ for each $i \in [t]$, our greedy procedure
stopped due to having $|I_i| = 32 \cdot \frac{c_\ell}{c^*}$ for each $i \in [t]$.
This contradicts the average size of the $I_i$'s being at most $8 \cdot \frac{c_\ell}{c^*}$.
So, every set $U_j$ must have been merged into some $V_i$. Thus, since
$|V_i| \leq \frac{n^*}{t}$ and the average of the $|V_i|$'s is $\frac{n^*}{t}$,
we have that for every $i \in [t]$, $|V_i| = \frac{n^*}{t}$. Moreover, for each $i \in [t]$,
$|I_i| \leq 8\frac{c_\ell}{c^*}$. Thus the $V_i$'s satisfy the lemma.

\smallskip

\noindent\textbf{Claim 2:} \textit{For every $i \in [t]$, $|V_{i'}| \geq \frac{r}{2}$.} \\
Suppose some index $i$ has $|V_{i}| < \frac{r}{2}$. Notice that, if
$|I_i| < 32 \cdot \frac{c_\ell}{c^*}$, then for every $i' \in [t]$,
$|V_{i'}| \leq |V_i| + \frac{r}{8} \leq \frac{3r}{4} \leq \frac{n^*}{t}$, i.e.,
contradicting Claim~1. Thus, $|I_i| = 32 \cdot \frac{c_\ell}{c^*}$ for each
$i \in [t]$ where $|V_i| < \frac{r}{2}$.
For each $i'\in [t], j'\in[\ell]$, let $I_{i'}^{j'}$ and $V^{j'}_{i'}$ be the
states of $I_{i'}$ and $V_{i'}$(respectively) directly after index $j'$ has
been added to some set $I_{i''}$ by the greedy algorithm.

We now let $\hat{j}$ be the largest index in $I_i$, and assume (without loss of
generality) that for every  $i' \in [t] \setminus \{i\}$,
if $|V^{\hat{j}}_{i'}| < \frac{r}{2}$, then $I^{\hat{j}}_{i'} < 32 \cdot \frac{c_\ell}{c^*}$.
Intuitively, $i$ is the ``first'' index which attains $|I_i| = 32 \cdot \frac{c_\ell}{c^*}$
while still having $|V_i| < \frac{r}{2}$.
Now, since $|I^{\hat{j}}_i| = 32 \cdot \frac{c_\ell}{c^*}$, and $|V^{\hat{j}}_i| < \frac{r}{2}$,
we have $|U_{\hat{j}}| < r \cdot \frac{c^*}{64c_\ell}$.
Thus, for every iteration $j> \hat{j}$, we have $|U_j| < r \cdot \frac{c^*}{64c_\ell}$.
This means that after iteration $\hat{j}$, the number of unallocated vertices is strictly less than:
\begin{center}
$\sum_{j=\hat{j}}^\ell U_j < \ell \cdot r \cdot \frac{c^*}{64c_\ell}
\leq t\cdot 8 \cdot \frac{c_\ell}{c^*} \cdot r \cdot \frac{c^*}{64c_\ell} = \frac{tr}{8}$.
\end{center}
In particular, this means that on average each set $V_i$ can grow by less than $\frac{r}{8}$.
However, due to our choice of $i$, we see that for every $i' \in [t]\setminus \{i\}$,
$|V^{\hat{j}}_{i'}| \leq |V^{\hat{j}}_{i}| + \frac{r}{8}< \frac{r}{2} + \frac{r}{8}$.
This means that even if we allocate all the remaining vertices, the average size of our
sets $V_i$ will be strictly less than  $\frac{3r}{4}$ $\leq \frac{n^*}{t}$, i.e., providing
a contradiction and proving Claim~2.

\smallskip

\noindent\textbf{Claim 3:} \textit{If every $j \in [\ell]$ is placed into
some $I_i$, the $V_i$'s satisfy the conditions of the lemma.} \\
First, note that $|I_i|$ is at most $32\cdot \frac{c_\ell}{c^*}$, i.e., $|N(V_i) \cap
\Sep| \in O(f(r))$. By Claim~2, we see that $|V_i| \geq \frac{n}{2}$ for each $i \in [t]$.
Additionally, from the greedy construction, we have that $|V_i| \leq \frac{n^*}{t} + \frac{r}{8}$.
Thus, $|V_i| \in [\frac{r}{2},\frac{9r}{8}] \subset [\frac{r}{2},2r]$.

\smallskip

We now describe the second step. By Claim~3, we assume there are unassigned sets $U_j$.
By Claim~2, for every $i \in [t]$, $|V_i| \geq \frac{r}{2}$.
Finally, by Claim~1, there is an index $i'$ where $|V_{i'}| > \frac{n^*}{t}$. 
Thus, since we have $t=\lceil\frac{n^*}{r}\rceil$ sets which partition
at most $n^*$ elements, there must be some index $i''$ where $|V_{i''}| \leq \frac{n^*}{t}$ and
$|I_{i''}| = 32\cdot \frac{c_\ell}{c^*}$,
i.e., $|U_{j^*}| \leq \frac{n^*}{t} \cdot (32\cdot \frac{c_\ell}{c^*})^{-1} \leq \frac{r\cdot c^*}{32\cdot c_\ell}$ where
$U_{j^*}$ is the largest unassigned set.
Notice that there are at most $\ell \leq t \cdot 8 \cdot \frac{c_\ell}{c^*}$
indices which can be assigned and all the remaining sets contain at most $|U_{j^*}|$
vertices. If we spread these remaining $U_j$'s uniformly throughout our $V_i$'s,
we will place at most $8\cdot \frac{c_\ell}{c^*} \cdot |U_{j^*}| \leq \frac{r}{4}$ vertices into each $V_i$.
Thus, for each $i \in [t]$, we have $|V_i| \leq \frac{n^*}{t} + \frac{r}{8} + \frac{r}{4} \leq 2r$.
So, by uniformly assigning these remaining indices, we have
$|I_i| \leq 40\cdot \frac{c_\ell}{c^*}$, $|V_i| \in [\frac{r}{2},2r]$,
and $|N(V_i) \cap \Sep| \leq 40\cdot \frac{c_\ell}{c^*} \cdot
c_1 f(\lfloor \frac{r}{8} \rfloor) \in O(f(r))$, as needed.

We conclude with a brief discussion of the time complexity. First, we generate
the $(\lfloor \frac{r}{8} \rfloor, c_1 f(\lfloor \frac{r}{8} \rfloor))$-division
in $h(n)$ time. We then sort the sets $|U_1| \geq \ldots \geq |U_\ell|$ (this can
be done in $O(n)$ time via bucket sort). In the next step we greedily fill the
index sets -- this takes $O(n)$ time. Finally, we place the remaining ``small''
sets uniformly throughout the $V_i$'s -- taking again $O(n)$ time. Thus,
we have $O(h(n) + n)$ time in total.
\end{proof}

We now prove a technical lemma which, together with the previous
lemma regarding uniform divisions, provides our uniform two-color balanced
divisions (see Theorem~\ref{thm:2color-uniform}) as discussed following this lemma.

\begin{lemma}\label{lem:partitioning_easy}
Let $c$ and $c'$ be positive constants, and $A=\{(a_1,b_1),\dots,(a_n,b_n)\}
\subseteq (\mathbb{Q}\cap[0,\infty))^2$ be a set of $2$-dimensional vectors where
$a_i+b_i \in [c',c]$ for each $i \in [n]$, and $\alpha \in [0,1]$ such that
$\sum_{i=1}^n a_i = \alpha\cdot \sum_{i=1}^n b_i$.
There is a permutation $p_1, \ldots, p_n$ of $[n]$ such that for any $1 \leq i \leq i' \leq n$, $|\sum_{j=i}^{i'} a_{p_j} - \alpha \cdot b_{p_j}| \leq 2\cdot c$.

Thus for any positive integer $q$, when $n$ is sufficiently larger than $q$,
there exist numbers $k\leq n$ and $q' \in [q,2q-1]$ and a partitioning of
$[n]$ into subsets $I_1,\dots,I_k$ such that for each $j\in [k]$ we have:

\begin{enumerate}[(i)]
\itemsep=0pt
\item $|I_j| \in \{q',q'+1\}$ (thus, $\sum_{i \in I_j} a_i+b_i \in [q'\cdot c', (q'+1)\cdot c]$), and
\item $|\sum_{i \in I_j} a_i - \alpha \cdot b_i| \leq 2\cdot c$.
\end{enumerate}

\noindent Moreover, the permutation $p_1, \ldots, p_n$ and partition can be computed in $O(n)$ time.
\end{lemma}
\begin{proof}
First, we partition $[n]$ into three sets $A_{> 0}$, $A_{<0}$, and $A_{=0}$
according to whether the \emph{weighted}
difference $d_i = a_i - \alpha\cdot b_i$ is positive, negative, or 0
(respectively). Note that, $\sum_{i=1}^{n} d_i = 0$ and for each $i \in [n]$, $|d_i| \leq c$.
We will pick indices one at a time from the sets $A_{> 0}$, $A_{< 0}$, $A_{=0}$ to form the desired permutation.

We now construct a permutation $p_1, \ldots, p_n$ on the indices $[n]$ so that any
consecutive subsequence $S$ has $|\sum_{i \in S} d_{p_i}| \leq 2\cdot c$.
For notational convenience, for each $j \in [n]$, we use $\delta_{<j}$ to denote $\sum_{i=1}^{j-1} d_{p_i}$. We now pick the $p_i$'s so that for each $j$, $|\delta_{<j}| \leq c$. We initialize $\delta_{<1}=0$.
For each $j$ from $1$ to $n$ we proceed as follows. Assume that
$|\delta_{<j}| \leq c$. We further assume that any index
$i \in \{p_1, \ldots, p_{j-1}\}$ has been removed from the sets $A_{>0}$,
$A_{<0}$, and $A_{=0}$.
If $\delta_{<j}$ is negative, $A_{> 0}$ must contain some index $j^*$ since $\sum_{i \in [n]} d_i =0$. Moreover, if we set $p_j= j^*$, we have $|\delta_{< j+1}| \leq c$ as needed (we also remove the index $j^*$ from $A_{>0}$ at this point).
Similarly, if $\delta_{<j}$ is positive, we pick any index $j^*$ from $A_{<0}$,
remove it from $A_{<0}$, and set $p_j = j^*$.
Finally, when $\delta_{<j}=0$), we simply take any index $j^*$ from $A_{>0} \cup  A_{<0} \cup A_{=0}$, remove it from $A_{>0} \cup A_{<0} \cup A_{=0}$, and set $p_j=j^*$.
Thus, in all cases we have $|\delta_{<j+1}| \leq c$.

Notice that, for any $1\leq j \leq j' \leq n$, we have $|\sum_{i=j}^{j'} d_{p_i}| = |\delta_{<j} - \delta_{<j'+1}|$ $\leq$ $|\delta_{<j}| + |\delta_{<j'+1}|$ $\leq 2 \cdot c$ (as needed for the first part of the lemma).

It remains to partition $[n]$ to form the sets $I_1, \ldots, I_k$.
This is accomplished
by splitting $p_1, \ldots, p_n$ into $t$ consecutive subsequences of almost equal size.
Namely, we pick $k = \lfloor \frac{n}{q} \rfloor$. We further let $z = n \mod q$, and $w = \lfloor \frac{z}{k} \rfloor$, $p = n - (q+w)\cdot k$.
From these integers, we make the sets $I_1,$ $\ldots,$ $I_{p}$ with
$q+w+1$ indices each and the sets $I_{p+1}, \ldots, I_k$ with $q+w$ indices each
by partitioning $\pi$ into these sets in order.  A simple calculation shows that
these sets satisfy the conditions of the lemma. Moreover, this construction
is clearly performed in $O(n)$ time.
\end{proof}

We will now use Lemmas~\ref{lem:uniform-(r,f)-division}~and~\ref{lem:partitioning_easy} to prove Theorem~\ref{thm:2color-uniform}. In particular, for a given two-colored graph $G$ where $G$ belongs to an $f$-separable graph class, we first construct a uniform $(r,c\cdot f(r))$-division $(\Sep, V_1, \ldots, V_t)$ of $G$ as in Lemma~\ref{lem:uniform-(r,f)-division}. From this division we can again carefully combine the $V_i$'s to make new sets $W_j$
where each $W_j$ has roughly the same size and contains roughly the same proportion of each color class as occurring in $G$. This follows by simply imagining each region $V_i$ of the uniform $(r,c\cdot f(r))$-division as a two-dimensional vector (according to its coloring) and then applying Lemma~\ref{lem:partitioning_easy}.

\begin{theorem}\label{thm:2color-uniform}
Let $\mathcal{G}$ be an $f$-separable graph class and $G=(V,E)$ be a
2-colored $n$-vertex graph in $\mathcal{G}$ with color classes $\Gamma_1,
\Gamma_2$ such that $|\Gamma_2| \geq |\Gamma_1|$. For any $q$ and $r \ll n$ where $r$ is suitably large, there
is an integer $t \in \Theta(\frac{n}{q\cdot r})$
such that $V$ can be partitioned into $t+1$ sets
$\Sep, V_1, \ldots, V_t$ where $c_1,c_2$ are constants independent
of our parameters $n,r,q$ and there is an integer $q' \in [q,2q-1]$ all satisfying
the following properties.
\begin{enumerate}[(i)]
  \item\label{item:sep-property} $N(V_i) \cap V_j = \emptyset$ for each $i\neq j$,
  \item\label{item:size-part} $|V_i| \in [\frac{q'\cdot r}{2},2\cdot (q'+1)\cdot r]$ for each $i$,
  \item\label{item:small-sep} $|N(V_i) \cap \Sep| \leq c_1 \cdot q \cdot f(r)$ for each $i$
  (thus, $|\Sep| \leq \sum_{i=1}^t|\Sep \cap N(V_i)| \leq \frac{c_2 \cdot f(r) \cdot n}{r}$).
  \item\label{item:balanced-part} $\left| |V_i \cap \Gamma_1| - \frac{|\Gamma_1|}{|\Gamma_2|} \cdot |V_i \cap \Gamma_2| \right| \leq 2\cdot r$
  \end{enumerate}

Moreover, such a partition can be found in $O(h(n) + n)$ time where $h(n)$
is the amount of time required to produce a uniform $(r,c\cdot f(r))$-division of
$G$.
\end{theorem}

\section{PTAS for $f$-Separable Maximum Coverage}\label{AnalysisPTAS}

In this section we formalize the notion of $f$-separable instances of the MC problem and prove our main result -- see Theorem~\ref{thm:main-thm}.

\begin{definition}\label{def:planarizable}
  A class $\mathcal{C}$ of instances of MC is called \emph{$f$-separable} if for any two disjoint feasible solutions $\mathcal{F}$ and $\mathcal{F}'$ of any instance in $\mathcal{C}$ there exists an $f$-separable graph $G$ with node set $\mathcal{F}\cup\mathcal{F}'$ with the following \emph{exchange} property.  If there is a ground element $u\in U$ that is covered both by $\mathcal{F}$ and $\mathcal{F}'$ then there exists an edge $(S,S')$ in $G$ with $S\in\mathcal{F}$ and $S'\in \mathcal{F}'$ with $u\in S\cap S'$.
\end{definition}

\begin{theorem}\label{thm:main-thm}
  Let $f\in o(n)$ be non-decreasing sublinear function. Then, any $f$-separable class of instances of MC that is closed under removing elements and sets admits a PTAS.
\end{theorem}

\begin{proof}
Our algorithm is based on local search. We fix a positive constant integer
$b\geq 1$.  Given an $f$-separable instance of MC, we pick an arbitrary
initial solution $\mathcal{A}$.  We check if it is possible to replace
$b$ sets in $\mathcal{A}$ with $b$ sets from $\mathcal{F}$ so that the
total number of elements covered is increased.  We perform such a
replacement (swap) as long as there is one. We stop if there is no
profitable swap and output the resulting solution.

In what follows, we show that for sufficiently large $b$ the
above algorithm yields a $(1-28c_1c_2f(b)/b)$-approximate
solution and that it runs in polynomial time (for
constant $b$). Here, $c_1$ and $c_2$ are the constants from
Theorem~\ref{thm:2color-uniform}. This will prove the claim of
Theorem~\ref{thm:main-thm} by letting $b$ sufficiently large.
Note that, if $c_1 < 1$, then we see that Theorem~\ref{thm:2color-uniform} 
also holds for $c_1 = 1$. Similarly, if $c_2 < 1$, then Theorem~\ref{thm:2color-uniform} 
also holds for $c_2 = 1$. Thus, we can safely assume that $c_1,c_2 \geq 1$. 

Since each step increases the number of covered elements, the number
of iterations of the above algorithm is at most $|U|$.  Each iteration
takes $O(k^b|\mathcal{F}|^{b})$ time.  Therefore, the total running
time of the algorithm is polynomial for constant $b$.

We now analyze the performance guarantee of the algorithm. To this
end, let $\mathcal{O}$ be an optimum solution to the instance and let
$\mathcal{A}$ be the (locally optimal) solution output by the
algorithm.  Let $\opt$, $\alg$ denote the number of elements covered
by $\mathcal{O}$, $\mathcal{A}$, respectively.

Suppose that $\alg<\left(1-\frac{28c_1c_2f(b)}{b}\right)\opt$. We want
to show that this would imply that there is a profitable swap as
this would contradict the local optimality of $\mathcal{A}$ and
hence complete the proof.

We claim that it suffices to consider the case when $\mathcal{O},\mathcal{A}$ are disjoint, which is justified as follows.  Assume that $\mathcal{O}\cap\mathcal{A}\neq\emptyset$.  We remove the sets in $\mathcal{O}\cap\mathcal{A}$ from $\mathcal{F}$ and all the elements covered by these sets from $U$.  Moreover, we decrease $k$ by $|\mathcal{O}\cap\mathcal{A}|$ and replace $\mathcal{O}$ with $\mathcal{O}\setminus\mathcal{A}$ and $\mathcal{A}$ with $\mathcal{A}\setminus\mathcal{O}$.  Since our class of instances is closed under removing sets and elements the resulting instance is still contained in the class.  Moreover, $|\bigcup\mathcal{A}|<\left(1-\frac{28c_1c_2f(b)}{b}\right)|\bigcup\mathcal{O}|$.  Finally, if we are able to show that there exists a feasible and profitable swap in the reduced instance then the same swap is also feasible and profitable in the original instance (with original solutions $\mathcal{A}$ and $\mathcal{O}$).

Therefore, we assume from now on that $\mathcal{A}$ and $\mathcal{O}$ are disjoint. Since our instance is $f$-separable, there exists an $f$-separable graph $G$
with precisely $2k$ nodes for the two feasible solutions $\mathcal{O}$
and $\mathcal{A}$ with the properties stated in
Definition~\ref{def:planarizable}.

We now apply our two colored separator theorem (Theorem~\ref{thm:2color-uniform}) to
$G$ with color classes $\Gamma_1=\mathcal{O}$ and $\Gamma_2=\mathcal{A}$ and
with parameters $r=b$ and $q=b$.

Since $|\mathcal{O}|=|\mathcal{A}|=k$, the two color classes in $G$
are perfectly balanced. Let $\mathcal{A}_i=\mathcal{A}\cap V_i$,
$\mathcal{O}_i=\mathcal{O}\cap V_i$,
$N_i^{\mathcal{O}}=N(V_i)\cap\Sep\cap\mathcal{O}$ and
$\bar{\mathcal{O}}_i=\mathcal{O}_i\cup N_i^{\mathcal{O}}$ for any part
$V_i$ with $i\in[t]$ of the resulting subdivision of $G$.

We can assume that every set in $\mathcal{O}$ is contained in $\bar{\mathcal{O}}_i$
for some $i\in [t]$. This can be achieved by suitably adding edges to $G$ while
maintaining the necessary properties of the uniform colored subdivision.  More
precisely, for every of the at most $c_2\cdot f(b)\cdot \frac{n}{b}$ many sets in
$\mathcal{X}\cap\mathcal{O}$ we add an edge to a set in $\mathcal{A}_i$
for some
$i\in[t]$. By Theorem~\ref{thm:2color-uniform}, we have $|V_i|\leq 4b^2$ and $|N(V_i)\cap\mathcal{X}|\leq c_1bf(b)\leq c_1b^2$ for each $i\in[t]$. Hence, we have that $t\geq \frac{n}{4c_1b^2}$.
Therefore, we can insert edges between the sets in $\mathcal{X}\cap\mathcal{O}$ and sets in $\mathcal{A}_i$, $i\in[t]$ so that
the neighborhood $N(V_i)\cap\mathcal{X}$ receives at most $4c_1c_2 f(b)b$ many
additional nodes for each $i\in [t]$. Note that the exchange property of
Definition~\ref{def:planarizable}
still holds as we only added edges.  Also the
properties of Theorem \ref{thm:2color-uniform} are still valid except that the
bound on the boundary size $|N(V_i)\cap\mathcal{X}|$ in
Property~(\ref{item:small-sep}) has increased to at most
$5c_1c_2\cdot b\cdot f(b)$ since $c_2\geq 1$.

The idea of the analysis is to consider for each $i\in [t]$ a feasible \emph{candidate
  swap} (called candidate swap $i$) that replaces in $\mathcal{A}$ the
sets $\mathcal{A}_i$ with some suitably chosen sets from
$\bar{\mathcal{O}}_i$.  We will show that if
$\alg<\left(1-\frac{28c_1c_2f(b)}{b}\right)\opt$ then at least one of the
candidate swaps is profitable leading to a contradiction.

To accomplish this, we will first show that there exists a profitable swap that replaces $\mathcal{A}_i$ with $\bar{\mathcal{O}}_i$.  This swap may be infeasible as $|\mathcal{A}_i|$ may be strictly smaller than $|\bar{\mathcal{O}}_i|$.  We will, however, show that a feasible and profitable swap can be constructed by adding only some of the sets in $\bar{\mathcal{O}}_i$.

For technical reasons we are going to define a set $Z$ of elements that we (temporarily) disregard from our calculations because they will remain covered and thus should not impact our decision which of the sets in $\bar{\mathcal{O}}_i$ we will pick for the feasible swap.  More precisely, let $Z=\{\,u\in A\cap B\mid A\in \mathcal{A}_i, B\in\mathcal{A}\setminus\mathcal{A}_i,i\in[t] \,\}$ be the set of elements that are covered by some $\mathcal{A}_i$ but that remain covered even if $\mathcal{A}_i$ is removed from $\mathcal{A}$.

Let $L_i=\bigcup\mathcal{A}_i\setminus Z$ be the set of elements that are ``lost'' when removing the $\mathcal{A}_i$ from $\mathcal{A}$.  Moreover,  let $W_i=\bigcup \bar{\mathcal{O}}_i\setminus \bigcup (\mathcal{A}\setminus\mathcal{A}_i)$ be the set of elements that are ``won'' when we add all the sets of $\bar{\mathcal{O}}_i$ after removing $\mathcal{A}_i$.

We claim that $\sum_{i=1}^t |L_i|\leq\alg-|Z|$. To this end, note that $Z\subseteq \bigcup\mathcal{A}$ and that the family $\{L_i\}_{i\in [t]}$ contains pairwise disjoint sets because all elements that are not exclusively covered by a single $\mathcal{A}_i$ are contained in $Z$ and thus removed.  On the other hand, we claim that $\sum_{i=1}^t|W_i|\geq\opt-|Z|$.  To see this, note first that every element in $Z$ contributes 0 to the left hand side and 0 or -1 to the right hand side.
Every element covered by $\mathcal{O}$ but not by $\mathcal{A}$ contributes at
least 1 to the left (because every set in $\mathcal{O}$ lies in some
$\bar{\mathcal{O}_i}$ by our extension of the exchange graph) hand side and
precisely 1 to the right hand side.
Finally, consider an element $u$ that is
covered both by $\mathcal{A}$ and by $\mathcal{O}$ but does not lie in $Z$.
This element lies in a set $S\in\mathcal{A}_i$ for some $i\in[t]$.  Because of
the definition of the exchange graph $G$ there is some set $T\in\mathcal{O}$
with $u\in T$ and some set $S'\in\mathcal{A}$ with $u\in S'$ such that $S'$
and $T$ are adjacent in $G$.  We have that $S'\in\mathcal{A}_i$, for, otherwise
$u\in Z$.  Because of the separator property of $\Sep$ (see
Property~(\ref{item:sep-property}) of Theorem~\ref{thm:2color-uniform}) we must
have $T\in\bar{\mathcal{O}}_i$.  Moreover $u$ lies in $W_i$ because it is not
contained in $Z$ but is covered by $\bar{\mathcal{O}}_i$.  Hence $u$
contributes at least 1 to the left hand side and precisely 1 to the right hand
side of  $\sum_{i=1}^t|W_i|\geq\opt-|Z|$, which shows the claim.

We have $\opt>|Z|$ and hence
\begin{displaymath}
  \min_{\substack{i\in[t]\\ |W_i|>0}}\frac{|L_i|}{|W_i|}\leq\frac{\sum_{i=1}^t|L_i|}{\sum_{i=1}^t|W_i|}\leq\frac{\alg-|Z|}{\opt-|Z|}\leq\frac{\alg}{\opt}<1-\frac{28c_1c_2f(b)}{b}\,.
\end{displaymath}

Hence, we can pick   $i\in[t]$ such that
\begin{equation}\label{eq:gain-infeas-swap}
\frac{|L_i|}{|W_i|}<1-\frac{28c_1c_2f(b)}{b}\,.
\end{equation}

Recall that $c_1\geq 1$ and assume that $b$ is large enough so that $f(b)\geq 1$. 
Then by Properties~(\ref{item:balanced-part}),~(\ref{item:size-part}), and the (due the addition of edges to $G$) modified Property~(\ref{item:small-sep}) of
Theorem~\ref{thm:2color-uniform}, we have that $||\mathcal{A}_i|-|\mathcal{O}_i||\leq 2b$, $|N(V_i)\cap\Sep|\leq 5c_1c_2\cdot b\cdot f(b)$, and $|V_i|\geq b^2/2$.  Because of $|\mathcal{A}_i|+|\mathcal{O}_i|=|V_i|$ this implies
$|\bar{\mathcal{O}}_i|\leq\frac12|V_i|+b+5c_1c_2b\cdot f(b)$ and $|\mathcal{A}_i|\geq\frac12|V_i|-b$. Hence

\begin{align}\label{eq:swap-cardinalities}
\begin{split}
  \frac{|\mathcal{A}_i|}{|\bar{\mathcal{O}}_i|} & \geq  \frac{\frac12|V_i|-b}{\frac12|V_i|+b+5c_1c_2b\cdot f(b)} \\
 & \geq   \frac{(\frac12|V_i|+b+5c_1c_2b\cdot f(b))-2b-5c_1c_2b\cdot f(b)}{\frac12|V_i|+b+5c_1c_2b\cdot f(b)} \\
& \stackrel{|V_i|\geq b^2/2}{\geq}   1-\frac{28c_1c_2f(b)}{b}\,.
\end{split}
\end{align}

\noindent We are now ready to construct our feasible and profitable swap.  
To this end let $Z_i=\bigcup(\mathcal{A}\setminus\mathcal{A}_i)$. We inductively define 
an order $S_1,\dots,S_{|\bar{\mathcal{O}}_i|}$ on the sets in $\bar{\mathcal{O}}_i$ where we
require that
\begin{displaymath}
S_j=\arg\max_{S\in\bar{\mathcal{O}}_i}\left|S\setminus\left(Z_i\cup\bigcup_{\ell=1}^{j-1}S_\ell\right)\right|
\end{displaymath}
for any $j=1,\dots,|\bar{\mathcal{O}_i}|$ where $S_1$ maximizes
$|S\setminus Z_i|$.

Consider the following process of iteratively building a set $W'$ starting with $W'=\emptyset$. Suppose
that we add to $W'$ the sets $(S_1\setminus Z_i),\dots,(S_{|\bar{\mathcal{O}}_i|}\setminus Z_i)$ in this order ending up with $W'=W_i$.
  In doing so, the incremental gain is monotonically decreasing due to the definition of the order on $\mathcal{O}_i$ and due to the submodularity of the objective function. Hence, for any prefix
of the first $j$ sets we have that
\begin{equation}\label{eq:prefix-gain}
\left|\left(\bigcup_{\ell=1}^jS_\ell\right)\setminus Z_i\right|\geq\frac{j\cdot|W_i|}{|\bar{\mathcal{O}}_i|}\,.
\end{equation}

Suppose that $|\bar{\mathcal{O}}_i|>|\mathcal{A}_i|$ (otherwise we can just add all sets in $\bar{\mathcal{O}}_i$).  Consider the swap where we replace the
$|\mathcal{A}_i|\leq b$ many sets $\mathcal{A}_i$ from the local
optimum $\mathcal{A}$ with at most $|\mathcal{A}_i|$ many sets
$\{S_1,\dots,S_{|\mathcal{A}_i|}\}$
from~$\bar{\mathcal{O}}_i$.

We now analyze how this swap affects the objective function value. By removing the sets in $\mathcal{A}_i$ the objective function value drops by
\begin{align*}
  |L_i| & \stackrel{\mathrm{(\ref{eq:gain-infeas-swap})}}{<} \left(1-\frac{28c_1c_2f(b)}{b}\right)\cdot |W_i|\\
 & \stackrel{\mathrm{~(\ref{eq:prefix-gain})}}{\leq} \left(1-\frac{28c_1c_2f(b)}{b}\right)\frac{|\bar{\mathcal{O}}_i|}{|\mathcal{A}_i|}\left|\left(\bigcup_{\ell=1}^{|\mathcal{A}_i|}S_\ell\right)\setminus Z_i\right|\\
 & \stackrel{\mathrm{~(\ref{eq:swap-cardinalities})}}{\leq} \left|\left(\bigcup_{\ell=1}^{|\mathcal{A}_i|}S_\ell\right)\setminus Z_i\right|\,.
\end{align*}
The right hand side of this inequality is the increase of the objective function due to adding the sets $\{S_1,\dots,S_{|\mathcal{A}_i|}\}$ after removing the sets in $\mathcal{A}_i$.

Therefore the above described swap is feasible and also profitable and thus $\mathcal{A}$ is not a local optimum leading to a contradiction.
\end{proof}

\section{Applications}\label{sec:Application}
In this section we describe several problems which are special instances of the MC problem. 
We then describe how a PTAS for each of these problems can be obtained from our analysis 
of local search (see Theorem~\ref{thm:Application}). 

\begin{problem}
  Let $H$ be a set of ground elements, $\mathcal{S}\subseteq 2^{H}$ be a set of ranges and $k$ be a positive integer.
   A range $S\in \mathcal{S}$ is hit by a subset  $H'$ of $H$ if $S\cap H'\neq\emptyset$. The \textsc{Maximum Hitting} (MH) problem asks for a $k$-subset $H'$ of $H$ such that the number  of ranges    hit by $H'$  is maximized.
\end{problem}

\begin{problem}
Let $G=(V,E)$ be a graph and $k$ be a positive integer. A vertex $v\in V$
dominates all the vertices adjacent to $v$ including $v$. The \textsc{Maximum
Dominating} (MD) problem asks for a $k$-subset $V'$ of $V$ such that
the number of vertices dominated by $V'$ is maximized.
 \end{problem}

\begin{problem}
Let $T$ be a 1.5D terrain  which is an $x$-monotone polygonal chain in the plane  consisting of a set of vertices $\{v_1,v_2,\ldots,v_m\}$ sorted in increasing order of their $x$-coordinate, and  $v_i$ and $v_{i+1}$ are connected by an edge for all $i\in[m-1]$.  For any two points $x,y\in T$, we say that $y$ guards $x$ if each point in $\overline{xy}$ lies above or on the terrain.   Given finite sets $X,Y\subseteq T$ and a positive integer $k$, the  \textsc{Maximum Terrain Guarding} (MTG) problem asks for a $k$-subset $Y'$ of $Y$ such that the number of points  of $X$ guarded by $Y'$ is maximized.
\end{problem}

Let $r$ be an even, positive integer. A set of regions in $\IR^2$, where each region is bounded by a closed Jordan curve, is called \emph{$r$-admissible} if for any two such regions $q_1,q_2$, the curves bounding them cross $s\leq r$ times for some even $s$ and $q_1\setminus q_2$ and $q_2\setminus q_1$ are connected regions.  A set of regions are called \emph{pseudo-disks} if is $2$-admissible.  For example each set of disks (of arbitrary size) and each set of squares (of arbitrary size) is a 2-admissible set and, as such, can be called pseudo-disks.

We now state the following theorem summarizing several consequences of Theorem~\ref{thm:main-thm}.
These results follow either from the corresponding SC problem being known to be \emph{planarizable} (that is, we can define an exchange graph as in Definition~\ref{def:planarizable} that is planar and thus $\sqrt{n}$-separable) or in case of claim $(D_2)$ and $(V)$ by construction the exchange graph as a minor of the input graph.
\begin{theorem}\label{thm:Application}
Local search gives a PTAS for the following classes of MC problems:
\begin{itemize}
\item[($C_1$)]  the set of ground elements is a set of points in $\IR^3$, and    the family   of subsets   is induced by a set of   half spaces in $\IR^3$

\item[($C_2$)] the set of ground elements is a set of points in $\IR^2$, and    the family of subsets  is induced by a set of  convex pseudodisks (a set of convex objects where any two objects can have at most two intersections in their boundary).
\end{itemize}

Local search gives a PTAS for the following MH problems:
  \begin{itemize}
  \item[($H_1$)]  the set    of ground elements is a set of points in $\IR^2$, and  the set   of ranges   is   induced by a set of $r$-admissible regions (this includes
  pseudodisks, same-height axis-parallel rectangles, circular disks,
  translates of convex objects).
  \item[($H_2$)]   the set   of ground elements is a set of points in $\IR^3$, and  the set  of ranges   is   induced by a set of half spaces in $\IR^3$.
  \end{itemize}

Local search gives a PTAS for MD problems in  each of the following graph classes:

\begin{itemize}
\item[($D_1$)] intersection graphs of homothetic copies of convex objects (which includes
  arbitrary squares, regular k-gons, translated and scaled copies
  of a convex object).
\item[($D_2$)] non-trivial minor-closed graph classes.

\end{itemize}

Additionally, the following problems admit a PTAS via local search
\begin{itemize}
\item[(V)] the MVC problem on $f$-separable and subgraph-closed graph classes,
\item[(T)] the MTG problem.
\end{itemize}

\end{theorem}
\begin{proof}[Proof of Theorem~\ref{thm:Application}]
  In what follows, we refer to several known results for SC where the
  respective instances are planarizable.  This always also implies
  that the corresponding MC problem is planarizable.  By the result of
  Mustafa and Ray~\cite{MustafaR09}, we know that the MC instance is
  planarizable when the family $\mathcal{F}$ is a set of half spaces
  in $\IR^3$, or a set of disks in $\IR^2$.  Recently, De and
  Lahiri~\cite{DeL16} showed that when the objects are convex
  pseudodisks, then the corresponding SC (and thus MC) instance is
  planarizable.  Thus, as a consequence of Theorem~\ref{thm:main-thm},
  we have ($C_1$) and ($C_2$).

  Note that MH problem is a special instance of MC problem, where the
  set $\mathcal{S}$ of ranges plays the role of $U$, and the set $H$
  plays the role of $\mathcal{F}$, where each set $h \in H$ contains
  all the range $S \in \mathcal{S}$ such that $S\cap h\neq \emptyset$.
  On the other hand, It follows from the result of Mustafa and
  Ray~\cite{MustafaR09} that an MH instance is planarizable when the
  set of ranges are (i) a set of $r$-admissible regions, (ii) set of half spaces in
  $\IR^3$.  Thus, Theorem~\ref{thm:main-thm} implies that an MH problem
  admits PTAS when the ranges are a set of $r$-admissible region or half
  spaces in $\IR^3$. Thus, we have ($H_1$) and ($H_2$).

  Observe that MD is a special instance of MC, where the set $V$ of
  vertices plays the role of the set $U$ of ground elements , and the
  family $\mathcal{F}$ consists of $|V|$ subsets of $V$ where each
  subset is corresponding to the set of vertices dominated by each
  vertex $v\in V$.  On the other hand, from the result of De and
  Lahiri~\cite{DeL16}, we know that corresponding instance of MD is
  planarizable when the graph $G$ is an geometric intersection graph
  induced by homothetic set of convex objects. Thus, as a consequence of
  Theorem~\ref{thm:main-thm}, we have ($D_1$).

  To prove ($D_2$), we claim that these MC instances are $f$-separable
  according to Definition~\ref{def:planarizable}. As noted before, each
  non-trivial minor-closed graph class is
  $(c\cdot\sqrt{n})$-separable~\cite{alonST1990minorclosed,robertson2004graph}
  (for a suitable constant $c$).
  Let $D, D'$ be two disjoint feasible solutions. We now construct an
  auxiliary graph $H$ as in Definition~\ref{def:planarizable}.
  We start with the node set
  $D\cup D'$ and an empty edge set.  Let $u\in V$ be node that is
  dominated by $D$ and by $D'$. If $u\in D$ then there is a neighbor
  $v\in D'$ of $u$. We add edge $uv$ to $H$.  The case $u\in D'$ is
  handled symmetrically.  If $u\notin D$ and $u\notin D'$ then there
  are neighbors $v\in D$ and $v'\in D'$ of $u$.  In this case we add
  $u$ and the two edges $uv$ and $uv'$ to $H$.  Note that the
  resulting graph is a subgraph of $G$. Now, we perform the following
  operation on $H$ as long as $H$ contains a node that is not in
  $D\cup D'$.  If such a node $u$ exists it must have precisely two
  neighbors $v\in D$ and $v'\in D'$ by construction of $H$.  We
  contract the edge $uv$ and identify the resulting node with $v$
  (lying in $D$).  As a result we obtain a minor $H$ of $G$ with node
  set $D\cup D'$.  It is easy to check that this graph satisfies the
  requirements of Definition~\ref{def:planarizable}.  Moreover,
  because $H$ is a minor of $G$, $H$ is also $(c\cdot\sqrt{n})$-separable.
  Thus, the $MD$ problem admits a PTAS on such graph classes.

  The proof of ($V$) is even simpler than the one of ($D_2$).  Let
  $D, D'$ be two disjoint feasible solutions. We are going to
  construct an auxiliary graph $H$ as in
  Definition~\ref{def:planarizable}.  We start with the node set
  $D\cup D'$ and an empty edge set.  For any edge $uu'$ that is
  covered by both $D$ and $D'$ we may assume $u\in D$ and $u'\in D'$.
  We add edge $uu'$ to $H$.  Note that the graph $H$ is a subgraph of
  $G$, and it fulfils the requirement of
  Definition~\ref{def:planarizable}.  Moreover, because $G$ is
  contained in a subgraph-closed, $f$-separable graph class we know
  that $H$ is $f$-separable and we obtain that $MVC$ problem admits a
  PTAS on such graph classes.

  It is easy to observe that MTG is a special instance of MC, where
  the set $X$ plays the role of $U$, and the set $Y$ plays the role of
  family $\mathcal{F}$ of subsets, where each $y\in Y$ contains all
  elements of $X$ which can be guarded by $y$.  On the other hand, we
  know from the result of Krohn et al.~\cite[Lem.\S 2]{KrohnGKV14}
  that MTG is planarizable. Thus, we have ($T$).
\end{proof}

\bibliographystyle{plain}
\bibliography{references}

\appendix

\section{Multicolored Separator Theorem}
\label{app:separators}

In this appendix we generalize our colored separator results to more than two color classes.

We employ the following $d$-dimensional vector partitioning result by Doerr and Srivastav
(Theorem~4.2 in~\cite{doerr2003multicolour}).
\begin{theorem}
\label{thm:d-dim-discrepancy}
Let $A=\{\mb{a}_1,\dots,\mb{a}_n\}\subseteq (\mathbb{Q}\cap[0,1])^d$
 be a set of $d$-dimensional vectors, let $k\leq n$ be a positive integer.  Then we can compute in polynomial
time a partition $I_1,\dots,I_k$ of $[n]$ into $k$ sets such that for any $j\in [k]$, we have that $\|\mb{\upmu}-\sum_{i\in I_j}\mb{a}_i\|_{\infty}\leq 2d$ where $\mb{\upmu}= \frac{1}{k} \cdot \sum_{i=1}^n\mb{a}_i$.\qed{}
\end{theorem}

We will now use the uniform $(r,f(r))$-division obtained in Lemma~\ref{lem:uniform-(r,f)-division} and combine this with the $d$-dimensional vector
partitioning theorem (Theorem~\ref{thm:d-dim-discrepancy}) to obtain a $d$-color uniform separator theorem on $f$-separable graph classes~(see Theorem~\ref{thm:d-col-uniform-(r,f)-division} below).
In particular, for a given $d$-colored graph $G$ where $G$ belongs to an $f$-separable graph class, we first construct a uniform $(r,c\cdot f(r))$-division $(\Sep, V_1, \ldots, V_\ell)$ of $G$ as in Lemma~\ref{lem:uniform-(r,f)-division}.
From this division we carefully combine the $V_i$'s to form sets $W_j$
where each $W_j$ has roughly the same size and contains roughly the same proportion of each color class as occurring in $G$.
The below theorem follows by imagining each region $V_i$ of the uniform $(r,c\cdot f(r))$-division as a $d$-dimensional vector
(whose coordinates correspond to the number of vertices of each color), scaling these vectors by $\frac{1}{2r}$, and then applying Theorem~\ref{thm:d-dim-discrepancy} with the parameter $k = \lceil\frac{\ell}{r}\rceil$ (i.e., $\Theta(\frac{n}{r^2})$) to obtain the result.

\begin{theorem}\label{thm:d-col-uniform-(r,f)-division}
Let $\mathcal{G}$ be an $f$-separable graph class and $G=(V,E)$ be a
$d$-colored $n$-vertex graph in $\mathcal{G}$ with color classes $Z_1, \ldots
Z_d$. For any $r\ll n$ where $r$ is suitably large%
, there
is an integer $t \in \Theta(\frac{n}{r^2})$
such that $V$ can be partitioned into $t+1$ sets
$\Sep, V_1, \ldots, V_t$ where $c_1,c_2$ are constants independent
of $n,r$, and the following properties are satisfied.
\begin{enumerate}[(i)]
  \item\label{item:sep-property-dcol} $N(V_i) \cap V_j = \emptyset$ for each $i\neq j$,
  \item\label{item:size-part-dcol} $|V_i| \in \Theta(r^2) + \Theta(dr)$ for each $i$,
  \item\label{item:small-sep-dcol} $|N(V_i) \cap \Sep| \leq c_1 \cdot r \cdot f(r)$ for each $i$
  (i.e., $|\Sep| \leq \sum_{i=1}^t|\Sep \cap N(V_i)| \leq \frac{c_2 \cdot f(r) \cdot n}{r}$).
  \item\label{item:balanced-part-dcol} $| |V_i \cap Z_q| - \frac{|Z_q|}{t} | \leq 4 \cdot r \cdot d$ for each $q \in [d]$.
  \end{enumerate}
Moreover, such a partition can be found in $O(h(n) + p(n))$ time where $h(n)$
is the amount of time required to produce a uniform $(r,c\cdot f(r))$-division of
$G$ and $p(n)$ is the polynomial running time of the vector partitioning in Theorem~\ref{thm:d-dim-discrepancy}.
\end{theorem}

The proof of Theorem~\ref{thm:d-dim-discrepancy} is algorithmic and uses an iterative LP rounding approach and thus has a high running time. We conclude this subsection by showing an alternate result using an algorithmic version of Steinitz Lemma with a quadratic running time (in $n$) at the expense of a slightly worse discrepancy bound.  This result can be used for an alternate version of Theorem~\ref{thm:d-col-uniform-(r,f)-division} with a better running time bound but slightly worse balancing bound.  In particular, we use Lemma~\ref{lem:AlgorithmicSteinitz} (below) to obtain a vector partition of $d$-dimensional vectors whose discrepancy is at most $3d+1$ with respect to the $L_\infty$-norm (see Lemma~\ref{lem:Discrepancy}).

\begin{lemma}\label{lem:AlgorithmicSteinitz}\cite[Th.1]{Barany1984} For a set $\{\mb{b}_1,\dots,\mb{b}_n\} \subseteq [-1,1]^d$
of $d$-dimensional vectors with $\sum_{i=1}^n\mb{b}_i=0$
in $O(n^2d^3 + nd^4)$ steps can be given
a permutation $\pi$ of the set $[n]$ such that for each $1\le l\le n$
$$\left\|\sum_{i=1}^l \mb{b}_{\pi(i)}\right\|_\infty\le\left\lfloor \frac 32d\right\rfloor.$$
\end{lemma}

\begin{lemma}\label{lem:Discrepancy}
Let $A=\{\mb{a}_1,\dots,\mb{a}_n\}\subseteq [0,1]^d$ be a set of $d$-dimensional vectors, let $k\leq n$ be a positive integer.
Then we can compute in $O(n^2d^3 + nd^4)$ time
a partition $I_1,\dots,I_k$ of $[n]$ into $k$ sets such
that for any $j\in [k]$, we have that $\|\mb{\upmu}-\sum_{i\in I_j}\mb{a}_i\|_{\infty}<2\left\lfloor \frac 32d\right\rfloor+1\le
3d+1$, where $\mb{\upmu}=\frac{1}{k} \cdot \sum_{i=1}^n\mb{a}_i$.
\end{lemma}
\begin{proof}
For each $i$, let $\mb{b}_i=\frac{k}{n}\cdot \mb{\upmu}-\mb{a}_i$.
Note that the set $\{\mb{b}_1,\dots,\mb{b}_n\}$ satisfies the conditions of
Lemma~\ref{lem:AlgorithmicSteinitz}.
Let $\pi$ be the permutation provided by this lemma.
Let $I'_1,\dots,I'_k$ be a partition of $[n]$ into $k$
consecutive discrete segments
such that sizes of any two elements of the partition differs by at most $1$.
For each $j\in [k]$, let $I_j=\pi(I'_j)$. We now have:

$$\left\|\mb{\upmu}-\sum_{i\in I_j}\mb{a}_i\right\|_{\infty}=
\left\|\mb{\upmu}-\sum_{i\in \pi(I'_j)}\mb{a}_i\right\|_{\infty}=
\left\|\mb{\upmu}+\sum_{k<j,\, i\in \pi(I'_k)}\mb{a}_i- \sum_{k\le j,\, i\in \pi(I'_k)}\mb{a}_i\right\|_{\infty}= $$
$$\left\|\mb{\upmu}-\frac{k\mb{\upmu}}{n}|I'_k|+
\sum_{k<j,\, i\in \pi(I'_k)}\left(\mb{a}_i-\frac{k\mb{\upmu}}{n}\right)-
 \sum_{k\le j,\, i\in \pi(I'_k)}\left(\mb{a}_i-\frac{k\mb{\upmu}}{n}\right)\right\|_{\infty}\le $$
$$\frac{k}n\|\mb{\upmu}\|_\infty\left|\frac nk-|I'_k|\right| +\left\lfloor \frac 32d\right\rfloor+\left\lfloor \frac 32d\right\rfloor<
\frac{k}n\|\mb{\upmu}\|_\infty +2\left\lfloor \frac 32d\right\rfloor\le
1+2\left\lfloor \frac 32d\right\rfloor.$$
\end{proof}

\end{document}